\documentclass[letterpaper, 10 pt, conference]{ieeeconf}  

\IEEEoverridecommandlockouts                              
\title{\LARGE \bf
Penalized Push-Sum Algorithm for Constrained Distributed Optimization with Application to Energy Management in Smart Grid}

\author{Tatiana Tatarenko, Jan Zimmermann, Volker Willert, J{\"u}rgen Adamy
\thanks{The authors are with the Control Methods and Robotics Lab at TU Darmstadt, Germany.
        }
\thanks{The work was gratefully supported by the German Research Foundation
(DFG) within the SPP 1984 ``Hybrid and multimodal energy systems: System theoretical methods for the transformation and operation of complex networks''.}
}

\usepackage{amsmath, amssymb, epsfig, graphicx, bm}
\usepackage{cases}
\usepackage{float, subfigure}
\usepackage{color}
\usepackage{amsthm}
\usepackage{amsfonts}
\usepackage{psfrag}
\usepackage{euscript}
\usepackage{cite, url}
\usepackage{breqn}
\usepackage[all,cmtip]{xy}
\usepackage{color,hyperref}
\definecolor{darkblue}{rgb}{0,0,1}
\hypersetup{colorlinks,breaklinks,
linkcolor=darkblue,urlcolor=darkblue,anchorcolor=darkblue,citecolor=darkblue}

\usepackage{tikz}
\usetikzlibrary{arrows}
\usetikzlibrary{positioning}
\usepackage{amsmath}
\usepackage{amssymb}
\usepackage{bm}
\usepackage[utf8]{inputenc}
\usetikzlibrary{arrows,shapes,snakes,automata,backgrounds,petri}

\newcommand{\EE}{\EuScript E}
\newcommand{\EN}{\EuScript N}

\theoremstyle{plain}
\newtheorem{lem}{Lemma}
\newtheorem{rem}{Remark}
\newtheorem{prop}{Proposition}
\newtheorem{theorem}{Theorem}
\newtheorem{cor}{Corollary}

\definecolor{amber}{rgb}{1.0, 0.49, 0.0}

\theoremstyle{definition}
\newtheorem{definition}{Definition}
\newtheorem{assumption}{Assumption}

\def\fb{\mathbf{f}}
\def\R{\mathbb{R}}

\def\Z{\mathbb{Z}}
\def\zb{\mathbf{z}}
\def\zx{\mathbf{x}}

\def\zq{\mathbf{q}}

\def\ab{\mathbf{a}}
\def\bb{\mathbf{b}}
\def\vb{\mathbf{v}}

\def\bp{\boldsymbol{p}}
\def\bx{\boldsymbol{x}}
\def\bv{\boldsymbol{v}}
\def\bmu{\boldsymbol{\mu}}
\def\bgamma{\boldsymbol{\gamma}}
\def\btheta{\boldsymbol{\theta}}
\def\bpsi{\boldsymbol{\psi}}

\def\zp{\mathbf{p}}
\def\bzx{\bar{\zx}}
\def\z*{\zb_t^*}
\def\la{\langle}
\def\ra{\rangle}

\begin{document}
\maketitle

\begin{abstract}
We study distributed convex constrained optimization on a time-varying multi-agent network. Each agent has access to its own local cost function, its local constraints, and its instant number of out-neighbors. The collective goal is to minimize the sum of the cost functions over the set of all constraints.
We utilize the \emph{push-sum protocol} to be able to solve this distributed optimization problem. We adapt the push-sum optimization algorithm, which has been studied in context of \emph{unconstrained} optimization so far, to convex \emph{constrained} optimization by introducing an appropriate choice of penalty functions and penalty parameters.
Under some additional technical assumptions on the gradients we prove convergence of the distributed penalty-based push-sum algorithm to the optimal value of the global objective function. We apply the proposed penalty-based push-sum algorithm to the problem of \emph{distributed energy management in smart grid} and discuss the advantages of this novel procedure in comparison with existing ones.
\end{abstract}

\section{Introduction}
Due to emergence of large-scaled networked systems with limited information, distributed multi-agent optimization problems have gained a lot of attention recently. In such systems a number of agents, represented by nodes over some communication graph, aim to optimize a global objective by taking only the local information into account. Beside the various applications of distributed optimization such as robust
sensor network control~\cite{Rabbat}, signal processing~\cite{touri2010infinite}, network routing~\cite{Neglia}, and machine learning~\cite{WillertS16,Tsianos2012}, an important and promising area of applicability is energy management of future smart grid~\cite{Zhao23, Zhao26, Zhao2017}. Smart grid is equipped with advanced communication technologies enabling efficient and distributed energy management between the grid's users \cite{Zhao2017, Zhao2}. However, from technical point of view, it is important to keep communication costs limited and choose a communication protocol that would require minimal coordination between the agents and stays robust against changes in the network topology \cite{Zhao9}. That is why, in this paper, we develop a communication-based optimization algorithm with the desired features mentioned above.

For this purpose we utilize the \emph{push-sum} communication protocol. This protocol was initially introduced in~\cite{16A2} and used in~\cite{Tsianos2012} for distributed optimization. The push-sum protocol is applicable to time-dependent network topology and it can overcome the restrictive assumptions on the communication graph structure such as double stochastic communication matrices~\cite{16A2, Tsianos2012}.
The work~\cite{A1} studied this algorithm over directed and time-varying communication in the case of well-behaved convex functions. The authors in~\cite{TAC2017} extended the results to a broader class of non-convex functions. However, all works on the push-sum algorithm presented in the literature so far dealt with unconstrained optimization. As in many applications, including energy management in smart grid, agents face a number of constraints. In this paper, we adapt the push-sum algorithm to the case of convex constrained optimization by introducing an appropriate choice of penalty functions and penalty parameters. Under some standard technical assumptions, we prove convergence of the resulting procedure to the optimal value of the system's objective function.

Another contribution of this paper consists in the application of the proposed procedure to the problem of energy management in smart grid.
In contrast to the recent communication-based procedure proposed in \cite{Zhao2017}, the penalty-based push-sum algorithm presented in this work is based on a \emph{time-dependent} directed communication topology with \emph{column-stochastic} matrices, where each agent merely needs to know the current number of its out-neighbors to define the elements of the communication matrix at each iteration. Moreover, it uses only one communication step per iteration, whereas the procedure in \cite{Zhao2017} requires two communication steps per iteration, with a row-stochastic and a column-stochastic communication matrix at the consequent communication iterations. Thus, the proposed penalty-based push-sum algorithm keeps communication costs cheaper and is able to adapt to the changes in the communication topology.

The paper is organized as follows. In Section~\ref{sec:ps} we introduce the penalty-based push-sum algorithm and prove its convergence. Section~\ref{sec:sg} deals with formulation of the general non-convex energy management problem in smart grid, presents its convex reformulation for which the penalty-based push-sum procedure can be applied, and demonstrates some simulation results. Section~\ref{sec:conclusion} concludes the paper.

\textbf{Notations.}
We will use the following notations throughout this paper: We denote the set of integers by $\Z$ and the set of non-negative integers by $\Z^+$. For the metric $\rho$ of a metric space $(X,\rho(\cdot))$ and two subsets $B_1\subset X$ and $B_2\subset X$, we let $\rho(B_1,B_2)=\max\{\sup_{x\in B_1}\inf_{y\in B_2}\rho(x,y), \sup_{y\in B_2}\inf_{x\in B_1}\rho(x,y)\}$. We denote the set $\{1,\ldots,n\}$ by $[n]$. We use boldface to distinguish between the vectors in a multi-dimensional space and scalars. We denote the dot product of two vectors $\ab$ and $\bb$ by $\la\ab,\bb\ra$. $\|\cdot\|$ denotes the standard Euclidean norm, whereas $\|\cdot\|_{l_1}$ is used to denote $l^1$-norm in the vector space. Throughout this work, all time indices such as $t$ belong to $\Z^+$. For vectors $\vb_i\in X^d$, $i\in[n]$, of elements in some vector space $X$ (over $\R$), we let $\bar{\vb}=\frac{1}{n}\sum_{i=1}^n\vb_i$. We say the function $F:\R^d\to\R$ be inf-compact, if the set $\{\zx\in\R^d : F(\zx) \le A \}$ is compact for
all $A \in \R$. The function $1_{\{A\}}(x)$ denotes the indicator of the set $A$ ($1_{\{A\}}(x)=1,$ if $x\in A$ and $1_{\{A\}}(x)
=0,$ otherwise).
The notation $o(x)$ as $x\to x_0$ is for some function $f(x)$ such that $\lim_{x\to x_0}\frac{f(x)}{x}=0$.

\section{Push-sum Algorithm for Distributed Constrained Optimization}\label{sec:ps}
In this section we adapt the push-sum algorithm to the case of  constrained convex optimization and prove convergence of the resulting procedure.
\subsection{Problem Formulation and Adapted Push Sum Algorithm}
 Let us consider the following general problem:
 \begin{align}\label{eq:const_opt}
  \min F(\zb)&=\sum_{i=1}^nF_i(\zb), \quad \zb\in\R^d,\cr
  \mbox{s.t. }c_1(\zb)&\le 0, \, c_2(\zb)\le0, \ldots,  c_n(\zb)\le 0, 
 \end{align}
 where $F_i:\R^d\to\R$, $c_i:\R^d\to\R$, $i=1,\dots,n$, are some differentiable convex functions. Let $\fb_i$ denote the gradient of the function $F_i$, $i\in[n]$, $\fb=\sum_{i=1}^{n}\fb_i$.
 This problem is formulated in a multi-agent system consisting of $n$ agents. Each agent $i$ has access to its local cost function $F_i$ and its local constraint described by the inequality $c_i(\zb)\le 0$\footnote{For the sake of notation simplicity, we assume that the local constraint of each agent $i$ is expressed by only one function $c_i$. The analysis below is applicable to problems, where agents have more than one constraint function.}.
 By the set $S$ we denote the set of solutions for \eqref{eq:const_opt}. By $F^*$ we denote the optimal value of the objective function $F$ in the problem \eqref{eq:const_opt}.

 At each time $t$, node $i$ can only communicate to  its  out-neighbors  in  some  directed  graph $G(t)$, where the graph $G(t)$ has the vertex  set $[n]$ and the edge set $E(t)$.
We use $N^{in}_i(t)$ and $N^{out}_i(t)$ to denote the in- and out-neighborhoods of node $i$ at time $t$. Each node $i$ is always considered to be an in- and out-neighbor of itself. We use $d_i(t)$ to denote the out-degree of node
$i$, and we assume that every node $i$ knows its out-degree at every time $t$.
The goal of the agents is to solve distributively the constrained minimization problem \eqref{eq:const_opt}.
We introduce the following standard definition for the sequence $G(t)$.
\begin{definition}\label{def:Sgraph}
	We say that a sequence of graphs $\{G(t)\}$ is \textit{$B$-strongly connected}, if, for any time $t\geq 0$, the graph
	\[G(t:t+B)=([n],E(t)\cup E(t+1)\cup \cdots \cup E(t+B-1)),\]
	is strongly connected. In other words, the union of the graphs over every $B$ time intervals is strongly connected.
\end{definition}
In the following analysis we assume that the sequence of the communication graphs $\{G(t)\}$ under consideration is $B$-strongly connected, which guarantees enough information ``mixing'' during communication between agents over time.

To deal with the problem described above, we aim to develop a distributed optimization procedure based on the push-sum protocol \cite{16A2}.
However, this protocol uses specific ratios of local agents' values to cancel out the effect
of information imbalances caused by limited agents' knowledge of their neighborhoods in time-dependent and directed communication network\cite{16A2, A1}.
That is why optimization methods based on projection onto the set of constraints cannot be applied here, as they violate the balance properties guaranteed by taking the corresponding ratio.
To overcome this limitation and to incorporate the constraints of the problem \eqref{eq:const_opt} into the optimization algorithm, we leverage the idea of penalty function methods \cite{PflugPenF}. We choose the convex penalty functions $\{\Psi_i(\zb)\}_i$ such that\footnote{Other candidates for penalty function can be found in \cite{Bertsekas}, Chapter 5.}
 \begin{align}\label{eq:penalty}
  \Psi_i(\zb)&=g(c_i(\zb))\cr
  g(u)&=\begin{cases}
        \log\left(\frac{e^u+e^{-u}}{2}\right), &\mbox{ if }u>0\\
        0, &\mbox{ if }u\le0.
       \end{cases}
 \end{align}
 Let $\Psi = \sum_{i=1}^{n}\Psi_i$, $\bpsi_i(\zb)$ denote the gradient of the function $\Psi_i(\zb)$, $\bpsi=\sum_{i=1}^{n}\bpsi_i$. Note that for each $i\in[n]$ the vector-function $\bpsi_i(\zb)$ is uniformly bounded over $\R^d$, given that $\nabla c_i(\zb)$ is uniformly bounded over $\R^d$.
 By adding the penalty function $\Psi$ to the objective function $F$ in \eqref{eq:const_opt}, we obtain the following unconstrained penalized optimization problem:
\begin{align}\label{eq:penprob}
 \min_{\zb\in\R^d} F_t(\zb)=F(\zb) + r_t\Psi(\zb) = \min_{\zb\in\R^d}\sum_{i=1}^{n}[F_i(\zb)+r_t\Psi_i(\zb)],
\end{align}
where $r_t$ is some positive penalty parameter. Note that as the functions $F_i$ and $\Psi_i$ are convex for all $i$ and $r_t>0$, the unconstrained problem above is convex. Let $S_t$ denote the set of solutions for \eqref{eq:penprob}. The connection between the penalized unconstrained problem \eqref{eq:penprob} and the initial constrained one in \eqref{eq:const_opt} is shown in the following proposition (see \cite{PflugPenF}):
\begin{prop}\label{pr:connection}
Let the function $F$ be inf-compact and $r_t\to\infty$ as $t\to\infty$. Then $S_t$ and $S$ are not empty and $S_t$ converges to $S$ as $t$ goes to infinity, namely $\lim_{t\to\infty}\rho(S,S_t)=0$. Moreover, $\lim_{t\to\infty}F_t^*=F^*$, where $F_t^* = \min_{\zb\in\R^d} F_t(\zb)$.
\end{prop}

Next, we apply the push-sum algorithm from \cite{A1} to the penalized problem \eqref{eq:penprob}.
We proceed with the formal algorithm formulation.
At every moment of time $t\in \Z^+$ each node $i$ maintains   vector   variables $\zb_i(t)$, $\zx_i(t)$, $\boldsymbol w_i(t)\in \mathbb{R}^d$, as  well  as  a  scalar  variable $y_i(t)$ such that $y_i(0)=1$ for all $i\in[n]$. These quantities are updated as follows:
\begin{subequations}\label{eq:constrps}
 \begin{align}
 \boldsymbol w_i(t+1)&=\sum_{j\in N^{in}_i(t)}\frac{\zx_j(t)}{d_j(t)},\label{eq:ps1}\\
 y_i(t+1)&=\sum_{j\in N^{in}_i(t)}\frac{y_j(t)}{d_j(t)},\label{eq:ps2}\\
 \zb_i(t+1)&=\frac{\boldsymbol w_i(t+1)}{y_i(t+1)},\label{eq:ps3}\\
 \zx_i(t+1)&=\boldsymbol w_i(t)-a_t[\fb_i(\zb_i(t+1))+r_t\bpsi_i(\zb_i(t+1))]\label{eq:ps4},
 \end{align}
\end{subequations}
where $a_t\ge 0$ is a time-dependent step size for all $t$.
The version of the push-sum algorithm above corresponds to the one proposed in \cite{A1}, where the optimization step \eqref{eq:ps4} is augmented by the penalty term $r_t\bpsi_i(\zb_i(t+1))$.
Note that the algorithm above is based on a time-dependent communication topology, where each agent $i$ merely needs to know its current out-degree $d_i(t)$ to follow the algorithm's steps.

\subsection{Convergence of the Algorithm}
In what follows, we analyze the convergence property of the algorithm \eqref{eq:constrps} under the following assumptions regarding the gradient functions.
\begin{assumption}\label{assum:bound_grad}
 The gradients $\fb_i$ and $\nabla c_i$  are uniformly bounded over $\R^d$ for all $i\in[n]$.
\end{assumption}
\begin{rem}\label{rem:bound}
 Since $\fb_i$ is assumed to be bounded for any $i\in[n]$, there exists a positive constant $L_i$ such that $\|\fb_i(\zb)\|\le L_i$ for any $\zb\in\R^d$. Let $L=\max_{i\in[n]}L_i$. Moreover, due to the bounded gradients $\nabla c_i$, the gradient function $\bpsi_i$ is bounded for any $i\in[n]$ (see \eqref{eq:penalty}). Thus, there exists a positive constant $M_i$ such that $\|\bpsi_i(\zb)\|\le M_i$ for any $\zb\in\R^d$. Let $M=\max_{i\in[n]}M_i$.
\end{rem}

\begin{assumption}\label{assum:Lipschitz}
 The gradients $\fb_i$ and $\nabla c_i$ are Lipschitz continuous over $\R^d$ for all $i\in[n]$.
\end{assumption}
\begin{rem}\label{rem:LC}
According to the choice of the penalty functions in \eqref{eq:penalty}, Lipschitz continuity of $\nabla c_i$ over $\R^d$ implies Lipschitz continuity of $\bpsi_i$ over $\R^d$ for all $i\in[n]$.
Thus, given Assumption~\ref{assum:Lipschitz}, there exist positive constants $l_i$ and $m_i$ such that $\|\fb_i(\zb_1)-\fb_i(\zb_2)\|\le l_i\|\zb_1-\zb_2\|$ and
 $\|\bpsi_i(\zb_1)-\bpsi_i(\zb_2)\|\le l_i\|\zb_1-\zb_2\|$ for any $\zb_1,\zb_2\in\R^d$ and all $i\in[n]$ respectively. Let $l=\max_{i\in[n]}l_i$ and $m=\max_{i\in[n]}m_i$.
\end{rem}

Moreover, we make the following assumption regarding the parameters $a_t$, $r_t$.
\begin{assumption}\label{assum:parameters}
 \begin{subequations}\label{eq:param}
  \begin{align}
   &a_t\le a_s \mbox{ for all $t\ge s$, } \, \sum_{t=0}^{\infty}a_t = \infty, \label{eq:a_t}\\
   &\qquad r_t\ge 1, \quad r_t\to \infty, \label{eq:r_t}\\
   &\sum_{t=0}^{\infty}a^2_t r^3_t < \infty, \quad r_{t+1}-r_t=o(a_t) \mbox{ as }a_t\to 0.  \label{eq:ar_t}
  \end{align}
 \end{subequations}
\end{assumption}

\begin{rem}
Note that the conditions \eqref{eq:a_t}, \eqref{eq:ar_t} imply that $a_t\to 0$ as $t\to\infty$.
 Appropriate sequences $\{a_t\}$ and $\{r_t\}$ that meet the assumption above can be, for example, $a_t = \frac{1}{t^{0.5+b}}$, $r_t=t^{0.25b}$,
 where $0<b<0.4$.
\end{rem}
According to the procedure \eqref{eq:constrps}, the running average of $\{\zx_i(t)\}_{i\in[n]}$, namely $\bar{\zx}(t)=\frac{1}{n}\sum_{i=1}^{n}\zx_i(t)$, fulfills the following iterations (see also \cite{A1}):
 \begin{align}\label{eq:constraverage}
 \bar{\zx}(t+1)=\bar{\zx}(t)-&\frac{a_t}{n}\sum_{i=1}^{n}[\fb_i(\zb_i(t+1))+r_t\bpsi_i(\zb_i(t+1))]\cr
 =\bar{\zx}(t)-&a_t(\fb(\bar{\zx}(t))+r_t\bpsi(\bar{\zx}(t)))\cr
 -&a_t\left[\frac1n\sum_{i=1}^{n}\fb_i({\zb_i(t+1)})-\fb(\bar{\zx}(t))\right]\cr
 -&a_tr_t\left[\frac1n\sum_{i=1}^{n}\bpsi_i({\zb_i(t+1)})-\bpsi(\bar{\zx}(t))\right].
 \end{align}
Some helpful results that will be used in the convergence analysis are presented in Appendix (see Theorems~\ref{th:th2} and \ref{th:th_nonnegrv}).
In particular, Theorem~\ref{th:th2}(a) implies that all $\zx_i(t+1)$, $i\in[n]$, converge with time to their running average $\bar{\zx}(t)$, given that Assumptions~\ref{assum:bound_grad} and \ref{assum:parameters} hold.

 Further we utilize the following notations:
 \begin{align*}
 \zq(t,\bar{\zx}(t))&= \frac{1}{n}\sum_{i=1}^{n}\fb_i({\zb_i(t+1)})-\fb(\bar{\zx}(t)),\cr
 \zp(t,\bar{\zx}(t))&= \frac{1}{n}\sum_{i=1}^{n}\bpsi_i({\zb_i(t+1)})-\bpsi(\bar{\zx}(t)).
 \end{align*}
We will use the following lemma which bounds the norms of the vectors $\zq(t,\bar{\zx}(t))$ and $\zp(t,\bar{\zx}(t))$ introduced above.
\begin{lem}\label{lem:bound}
Let Assumptions~\ref{assum:bound_grad}-\ref{assum:parameters} hold. Then, there exists $q(t)$ and $p(t)$ such that the following holds for the process \eqref{eq:constrps}:
 \begin{align*}
 \|\zq(t,\bar{\zx}(t))\|&\le \frac{l}{n}\sum_{i=1}^{n}\|\zb_i(t+1)-\bar{\zx}(t)\|= q(t),
 \end{align*}
  \begin{align*}
 \|\zp(t,\bar{\zx}(t))\|&\le \frac{m}{n}\sum_{i=1}^{n}\|\zb_i(t+1)-\bar{\zx}(t)\|= p(t),
 \end{align*}
 such that
 \begin{align*}
 \sum_{t=0}^{\infty} a_tr^3_tq(t)<\infty, \quad \sum_{t=0}^{\infty} a_tr^3_tp(t)<\infty.
\end{align*}
\end{lem}
\begin{proof}
 Due to Assumption~\ref{assum:Lipschitz},
 \begin{align*}
  \|\zq(t,\bar{\zx}(t))\|& = \|\frac{1}{n}[\sum_{i=1}^{n}\fb_i({\zb_i(t+1)})-\sum_{i=1}^{n}\fb_i(\bar{\zx}(t))]\|\cr
  &\le \frac{l}{n}\sum_{i=1}^{n}\|\zb_i(t+1)-\bar{\zx}(t)\|.
 \end{align*}
 Let $q(t) = \frac{l}{n}\sum_{i=1}^{n}\|\zb_i(t+1)-\bar{\zx}(t)\|$.
Next, let us consider the series
\[\sum_{t=1}^{\infty}b_ta_t\|\fb_i(\zb_i(t+1))+r_t\bpsi_i(\zb_i(t+1))\|_1.\]
If $b_t = a_tr_t^3$,
\[\sum_{t=1}^{\infty}b_ta_t\|\fb_i(\zb_i(t+1))+r_t\bpsi_i(\zb_i(t+1))\|_1<\infty,\]
as $\fb_i(\zb_i(t+1))$ and $\bpsi_i(\zb_i(t+1))$ are bounded, \eqref{eq:r_t} and \eqref{eq:ar_t} hold for $a_t$ and $r_t$.
Thus, we can use Theorem~\ref{th:th2} from Appendix to conclude that
 \begin{align*}
  \sum_{t=0}^{\infty}a_tr^3_tq(t) =l\sum_{t=0}^{\infty}\frac{a_tr_t^3}{n}\sum_{i=1}^{n}\|\zb_i(t+1)-\bzx(t)\| & <\infty.
 \end{align*}
Analogously, one can show that  $\sum_{t=0}^{\infty} a_tr^3_tp(t)<\infty$.
\end{proof}

\begin{rem}\label{rem:forLemma}
 Note that due to the choice of the parameters in Assumption~\ref{assum:parameters} and the fact that under Assumptions~\ref{assum:bound_grad}-\ref{assum:parameters} both $p(t)$ and $q(t)$ tend to 0 as $t\to\infty$ (see Theorem~\ref{th:th2} in Appendix),
$\sum_{t=0}^{\infty} a_t^{q_1}r^{q_2}_tq^2(t)<\infty$, $\sum_{t=0}^{\infty} a_t^{p_1}r^{p_2}_tp^2(t)<\infty$, and $\sum_{t=0}^{\infty} a_t^{2}r^2_tp(t)q(t)<\infty$
for all integers $q_1,p_1\ge 1$, $q_2,p_2\in\{0,1,2,3\}$.
\end{rem}

Now we state the main result for the penalty-based push-sum algorithm~\eqref{eq:constrps}.
\begin{theorem}\label{th:main_ps}
 Let the function $F$ in the problem \eqref{eq:const_opt} be inf-compact. Let Assumptions~\ref{assum:bound_grad}-\ref{assum:parameters} hold.  Then  all local variables $\zb_i(t+1)$, $i\in[n]$, in the procedure \eqref{eq:constrps} reach a consensus as $t\to\infty$ and each limit point of this consensus corresponds to a solution to the problem \eqref{eq:const_opt}, given that the sequence of the communication graphs $\{G(t)\}$ under consideration is $B$-strongly connected.
\end{theorem}
\begin{proof}
First, we will show that $\lim_{t\to\infty} \Psi(\bar{\zx}(t)) = 0$.
In particular, it will mean that all limit points of $\{\bar{\zx}(t)\}$ belong to the feasible set $C = \{\zb\in\R^d \,|\, c_i(\zb)\le 0, i\in[n]\}$.
Taking the Mean-value Theorem and relation \eqref{eq:constraverage} into account, and using the notation
\begin{align}\label{eq:fb}
 \tilde{\fb}(t,\bar{\zx}(t))=\fb(\bar{\zx}(t))+r_t\bpsi(\bar{\zx}(t))+\zq(t,\bar{\zx}(t))+r_t\zp(t,\bar{\zx}(t)),
\end{align}
 we get that for some $\theta\in[0,1]$ and $\tilde{\zx}(t)=\bar{\zx}(t)-\theta a_t\tilde{\fb}(t,\bar{\zx}(t))$
\begin{align}\label{eq:genop}
  \Psi&(\bar{\zx}(t+1))=\Psi(\bar{\zx}(t))-a_t\la\bpsi(\bar{\zx}(t)),\tilde{\fb}(t,\bar{\zx}(t))\ra\cr
  &+a_t[\la\bpsi(\bar{\zx}(t)),\tilde{\fb}(t,\bar{\zx}(t))\ra-\la\bpsi(\tilde{\zx}(t)),\tilde{\fb}(t,\bar{\zx}(t))\ra].
  \end{align}
According to Remark~\ref{rem:LC} and the Cauchy-Schwarz inequality, we obtain
\begin{align}\label{eq:genop1}
  \la\bpsi(\bar{\zx}(t)),&\tilde{\fb}(t,\bar{\zx}(t))\ra-\la\bpsi(\tilde{\zx}(t)),\tilde{\fb}(t,\bar{\zx}(t))\ra\cr
  &\le\|\tilde{\fb}(t,\bar{\zx}(t))\|\|\bpsi(\bar{\zx}(t))-\bpsi(\tilde{\zx}(t))\|\cr
  &\le\|\tilde{\fb}(t,\bar{\zx}(t))\|mn\|\bar{\zx}(t)-\tilde{\zx}(t)\|\cr
  &\le \theta mn a_t\|\tilde{\fb}(t,\bar{\zx}(t))\|^2,
  \end{align}
  where $m$ is the constant defined in Remark~\ref{rem:LC}.

Next, using Remark~\ref{rem:LC}, Lemma~\ref{lem:bound}, and due to the Cauchy-Schwarz inequality, we get
\begin{align}\label{eq:genop2}
  &\|\tilde{\fb}(t,\bar{\zx}(t))\|^2\le\|\fb(\bar{\zx}(t))\|^2 \cr
  &\,+ r^2_t\|\bpsi(\bar{\zx}(t))\|^2 + \|\zq(t,\bar{\zx}(t))\|^2 + r_t^2\|\zp(t,\bar{\zx}(t))\|^2\cr
  &\, + 2\|\fb(\bar{\zx}(t))\|r_t\|\bpsi(\bar{\zx}(t))\| + 2\|\fb(\bar{\zx}(t))\|\|\zq(t,\bar{\zx}(t))\| \cr
  &\, + 2 r_t\|\zq(t,\bar{\zx}(t))\|\|\zp(t,\bar{\zx}(t))\| + 2r_t\|\bpsi(\bar{\zx}(t))\|\|\zq(t,\bar{\zx}(t))\| \cr
  &\, + 2r_t^2\|\bpsi(\bar{\zx}(t))\|\|\zp(t,\bar{\zx}(t))\| + 2\|\fb(\bar{\zx}(t))\|\|\zp(t,\bar{\zx}(t))\| \cr
  &\le k_1(\frac12 + q(t) + r_t p(t)) + k_2 (\frac12 r_t^2 + r_tq(t) + r_t^2p(t))  \cr
  &\,+ k_3 r_t + q^2(t)+ r_t^2p^2(t) + 2r_t p(t)q(t) =g_0(t)
\end{align}
for some positive constants $k_1, k_2,$ and $k_3$.
Finally, using the definition of $\tilde{\fb}(t,\bar{\zx}(t))$ in \eqref{eq:fb} and the Cauchy-Schwarz inequality again, we obtain for some positive constants $k_4$ and $k_5$ that
\begin{align}\label{eq:genop3}
\la\bpsi(\bar{\zx}(t)),&\tilde{\fb}(t,\bar{\zx}(t))\ra\ge\la\bpsi(\bar{\zx}(t)),\fb(\bar{\zx}(t))\ra+r_t\|\bpsi(\bar{\zx}(t))\|^2\cr
&\quad-k_4\|\zq(t,\bar{\zx}(t))\|-k_5r_t\|\zp(t,\bar{\zx}(t))\|\cr
&\ge  \|\bpsi(\bar{\zx}(t))\|(-\|\fb(\bar{\zx}(t))\|+r_t\|\bpsi(\bar{\zx}(t))\|)\cr
&\quad-k_4q(t)-k_5r_tp(t)\cr
&\ge \|\bpsi(\bar{\zx}(t))\|1_{\{\|\bpsi(\bar{\zx}(t))\|\ne 0\}}-k_4q(t)-k_5r_tp(t),
\end{align}
where the first inequality is due to the bounded $\|\bpsi\|$ and the last inequality is due to the fact that $0<r_t\to\infty$ as $t\to\infty$ and  $\|\fb\|$ are bounded.

By substituting \eqref{eq:genop1}-\eqref{eq:genop3} to \eqref{eq:genop} we can write:
\begin{align*}
  \Psi(\bar{\zx}&(t+1))\cr
  &\le\Psi(\bar{\zx}(t)) -a_t\|\bpsi(\bar{\zx}(t))\|1_{\{\|\bpsi(\bar{\zx}(t))\|\ne 0\}} + g(t),
 \end{align*}
where $g(t)=a^2_t(g_0(t)+k_4q(t)+k_5r_tp(t))$. According to the choice of $a_t$ and $r_t$, Lemma~\ref{lem:bound}, and Remark~\ref{rem:forLemma}, we obtain that $\sum_{t=0}^{\infty}g(t)<\infty$. Thus, using the well-known result on the sequences of non-negative variables presented in Theorem \ref{th:th_nonnegrv} (see Appendix),
we conclude that  $\lim_{t\to\infty}\Psi(\bar{\zx}(t))$ exists, is finite, and $\sum_{t=0}^{\infty}a_t\|\bpsi(\bar{\zx}(t))\|^2<\infty$. Thus, $\liminf_{t\to\infty} \|\bpsi(\bar{\zx}(t)\| = 0$, since $\sum_{t=0}^{\infty}a_t=\infty$ (see \eqref{eq:a_t}).
It implies existence of a subsequence $\{t_k\}\subseteq\{t\}$ such that $\lim_{k\to\infty} \|\bpsi(\bar{\zx}(t_k)\| = 0$ and, as $\|\bpsi(\zb)\| = 0$ if and only if $\Psi(\zb) = 0$, we conclude that $\lim_{k\to\infty} \Psi(\bar{\zx}(t_k)) = 0$. Thus,
\begin{align}\label{eq:limPsi}
 \lim_{t\to\infty} \Psi(\bar{\zx}(t)) = 0.
\end{align}
Next, let us notice that $F(\bar{\zx}(t+1))=F_t(\bar{\zx}(t+1))-r_t\Psi(\bar{\zx}(t+1))=F_{t+1}(\bar{\zx}(t+1))-r_{t+1}\Psi(\bar{\zx}(t+1))$. Hence, taking into account that $\Psi(\bar{\zx}(t))=o(1)$  as $t\to\infty$ (see \eqref{eq:limPsi}), $r_{t+1}-r_t=o(a_t)$, $a_t\to 0$ as $t\to\infty$, and using Mean-value Theorem, we get
\begin{align}\label{eq:F_t}
&F_{t+1}(\bar{\zx}(t+1)) = F_t(\bar{\zx}(t+1))+(r_{t+1}-r_t)\Psi(\bar{\zx}(t+1))\cr&
=F_t(\bar{\zx}(t)-a_t\tilde{\fb}(t,\bar{\zx}(t)))+o(a_t)\cr
&=F_t(\bar{\zx}(t))-a_t\la\nabla F_t(\bar{\zx}(t)),\tilde{\fb}(t,\bar{\zx}(t))\ra+o(a_t)\cr
&+a_t\la\nabla F_t(\bar{\zx}(t))-\nabla F_t({\zx}'(t)),\tilde{\fb}(t,\bar{\zx}(t))\ra,
\end{align}
where ${\zx}'(t) = \bar{\zx}(t)-\beta a_t\tilde{\fb}(t,\bar{\zx}(t))$ for some $\beta\in[0,1]$.

According to Assumption~\ref{assum:Lipschitz}, there exists some $l_1>0$ such that
\begin{align*}
 \la\nabla F_t(\bar{\zx}(t))-\nabla F_t({\zx}'(t)),\tilde{\fb}(t,\bar{\zx}(t))\ra\\
 \le l_1 a_tr_t\|\tilde{\fb}(t,\bar{\zx}(t)))\|^2\le l_1 a_tr_tg_0(t),
\end{align*}
where for the first inequality we used the Cauchy-Schwarz inequality and the last inequality is due to \eqref{eq:genop2}.
Hence, due to the Cauchy-Schwarz inequality and the definition of $\tilde{\fb}(t,\bar{\zx}(t))$ in \eqref{eq:fb}, we obtain from \eqref{eq:F_t} that
\begin{align*}
F_{t+1}&(\bar{\zx}(t+1))\le F_t(\bar{\zx}(t))\\
&-a_t\|\nabla F_t(\bar{\zx}(t))\|^2-a_t\la\fb(\bar{\zx}(t)),\zq(t,\bar{\zx}(t))\ra\\
&-a_tr_t\la\bpsi(\bar{\zx}(t)),\zq(t,\bar{\zx}(t))\ra-a_tr_t\la\fb(\bar{\zx}(t)),\zp(t,\bar{\zx}(t))\ra\\
&-a_tr_t^2\la\bpsi(\bar{\zx}(t)),\zp(t,\bar{\zx}(t))\ra+l_1a^2_tr_tg_0(t)+o(a_t)\\
&\qquad\qquad\le F_t(\bar{\zx}(t))-a_t(\|\nabla F_t(\bar{\zx}(t))\|^2+o(1))\\
&+a_tq(t)\|\fb(\bar{\zx}(t))\|+a_tr_tq(t)\|\bpsi(\bar{\zx}(t))\|\\
&+a_tr_tp(t)\|\fb(\bar{\zx}(t))\|+a_tr^2_tp(t)\|\bpsi(\bar{\zx}(t))\|+l_1a^2_tr_tg_0(t)\\
&\qquad\quad= F_t(\bar{\zx}(t))-a_t(\|\nabla F_t(\bar{\zx}(t))\|^2+o(1))+g_1(t),
\end{align*}
where
\begin{align*}
g_1(&t)=a_tq(t)\|\fb(\bar{\zx}(t))\|+a_tr_tq(t)\|\bpsi(\bar{\zx}(t))\|\cr
&+a_tr_tp(t)\|\fb(\bar{\zx}(t))\|+a_tr^2_tp(t)\|\bpsi(\bar{\zx}(t))\|+l_1a^2_tr_tg_0(t).
\end{align*}
Thus, $\sum_{t=0}^{\infty}g_1(t)<\infty$, due to Assumption~\ref{assum:bound_grad}, the choice of $a_t$ and $r_t$, and Lemma \ref{lem:bound} (see Remark~\ref{rem:forLemma}).
Hence, according to Theorem~\ref{th:th_nonnegrv} from Appendix, we can conclude that
\begin{align}\label{eq:limitF_t}
 \mbox{$F_t(\bar{\zx}(t))$ has a limit as $t\to\infty$.}
\end{align}
 Moreover,
$\sum_{t=1}^{\infty}a_t\|\nabla F_t(\bar{\zx}(t))\|^2<\infty$, which, due to \eqref{eq:a_t}, implies $\liminf_{t\to\infty} \|\nabla F_t(\bar{\zx}(t))\|= 0$.
Let us choose a subsequence $\{t_k\}\subseteq\{t\}$ such that $\lim_{k\to\infty} \|\nabla F_{t_k}(\bar{\zx}(t_k))\|= 0$. Due to convexity of $F_t$ over $\R^d$ for all $t$, the last limit implies that
$\lim_{k\to\infty}[F_{t_k}(\bar{\zx}(t_k))-F^*_{t_k}] = 0$,
where $F^*_{t_k} = \min_{\bx\in\R^d} F_{t_k}(\bx)$. Next, due to Proposition~\ref{pr:connection} and as $r_t\to\infty$, we conclude that $\lim_{k\to\infty}F_{t_k}(\bar{\zx}(t_k))-F^* = 0$, which together with \eqref{eq:limitF_t} implies that $\lim_{t\to\infty}F_t(\bar{\zx}(t))=F^*$  and, hence, every limit point of $\bar{\zx}(t)$ is a solution to the problem \eqref{eq:const_opt}. Finally, by invoking Theorem~\ref{th:th2}(a), we conclude the result.
\end{proof}
\begin{rem}\label{rem:strictconv}
 Note that if, additionally to the conditions in Theorem~\ref{th:main_ps}, the function $F$ is assumed to be strictly convex, then there exists a unique solution $\zb^*$ to the problem \eqref{eq:const_opt}. In this case, Theorem~\ref{th:main_ps} implies convergence of all $\zb_i(t)$ evolving according to the algorithm \eqref{eq:constrps} to this optimum $\zb^*$.
\end{rem}

\section{Applications in Energy Management}\label{sec:sg}
We consider a problem of energy management formulated and analyzed in \cite{Zhao2017}.
Let $\EN_g$ and $\EN_d$ be the sets of distributed generators and responsive demands in a power grid with $N_g$ and $N_d$ elements respectively.
Let $N = N_g+N_d$. A directed
connected time-dependent graph $G(t) = ([N], \EE(t))$ is used to represent the communication topology of the network in the grid, where $[N] = \EN_g\cup \EN_d$ is the set of the nodes containing generators and demands and $\EE(t) \subseteq [N]\times[N]$ is the edge set. Note that $(j,i) \in \EE(t)$ if and only
if the node $i\in[N]$ can receive information from node $j\in[N]$ at time $t$.
In contrast to the previous works \cite{Asr2014, Zhao2017}, in this paper we focus on a broader class of the communication topology containing time-dependent graphs and requiring each user to know only its current out-degree to construct an appropriate communication matrix.
For this purpose, we will apply the penalty-based push-sum algorithm introduced and analyzed in Section~\ref{sec:ps} to the distributed optimization formulated below.

We consider the following generation and demand capacities in the system:
\[p_i\in [p_i^m, P_i^M], \, i\in \EN_g, \quad p_j\in [p_j^m, P_j^M], \, j\in \EN_d.\]
The cost function $C_i: \R\to\R$ of each generator $i\in\EN_g$ is:
\begin{align}\label{eq:costfun}
 C_i(p_i) = \begin{cases}
             a_ip_i^2 + b_ip_i + c_i, \, &\mbox{if }p_i\in [p_i^m, P_i^M],\cr
             (2a_ip_i^m + b_i)p_i, \, &\mbox{if }p_i\le p_i^m,\cr
             (2a_ip_i^M + b_i)p_i, \, &\mbox{if }p_i\ge p_i^M.
            \end{cases}
\end{align}
where $a_i, b_i, c_i$ are positive fitting parameters. Thus, the cost functions are strongly convex functions.

The utility function $U_j:\R\to\R$ of each demand $j\in\EN_d$ has the following properties:
\begin{align}\label{eq:utfun}
 &U_j(0)=0, \quad \frac{d U_j}{d p_j}>0 \,\mbox{ (non-decreasing)},\cr
 &\exists K_1,K_2: \frac{d U_j}{d p_j}<K_1, \cr
 & \quad K_2\le\frac{d^2U_j}{(d p_j)^2}\le 0 \,\mbox{ (get saturated)}.
\end{align}
Thus, the utility functions are concave functions.

Let $\bp\in\R^N$ be the vector with coordinates $p_i$, $i\in[N]$.
The goal in the power grid is to solve distributively the following energy management problem\footnote{For more details on the problem formulation see \cite{Zhao2017}.}:
\begin{equation}\label{eq:problem}
 \min_{\bp} \, \sum_{i\in \EN_g} C_i(p_i) - \sum_{j\in\EN_d} U_j(p_j)
 \end{equation} \addtocounter{equation}{-1}\vspace{-0.3cm}
 \begin{subequations}
    \begin{align}
s.t. \, &\sum_{i\in \EN_g}(p_i - l_ip_i^2) = \sum_{j\in\EN_d} p_j\label{eq:pr2}\\
 &p_i^m\le p_i\le P_i^M, \quad i\in \EN_g\label{eq:pr3}\\
 &p_j^m\le p_j\le P_j^M, \quad j\in \EN_d\label{eq:pr4}.
\end{align}
\end{subequations}
In the problem above the constraint \eqref{eq:pr2} corresponds to the balance between the generated and the demanded power in the network, where each parameter $l_i$, $i\in\EN_g$, corresponds to the coefficient of the transmission losses induced by the generator $i$ and satisfies $0\le l_i <a_i$.
Note that the problem~\eqref{eq:problem} is non-convex due to the non-convex constraint defined by  \eqref{eq:pr2}.

\subsection{Problem reformulation with constraints based on local information}
To implement the distributed penalized push-sum algorithm to the energy management problem~\eqref{eq:problem}, we need to find its appropriate convex reformulation such that any solution to this reformulation provides a solution for \eqref{eq:problem}. Moreover, as the constraint \eqref{eq:pr2} contains the information on the ``loss'' parameter $l_i$ of each $i\in\EN_g$, we aim to find a reformulation, where no constraint requires knowledge about the local properties of other nodes in the network.
First of all, let us notice that the problem~\eqref{eq:problem} is equivalent to the following one:
\begin{equation}\label{eq:problemRef0}
 \min_{\bp,\bv} \, \sum_{i\in \EN_g} C_i(p_i) - \sum_{j\in\EN_d} U_j(p_j)
 \end{equation} \addtocounter{equation}{-1}\vspace{-0.3cm}
 \begin{subequations}
    \begin{align}
s.t. \, &\sum_{i\in \EN_g}(p_i - v_i) = \sum_{j\in\EN_d} p_j\label{eq:prRef2}\\
 &p_i^m\le p_i\le P_i^M, \quad i\in \EN_g\label{eq:prRef3}\\
 &p_j^m\le p_j\le P_j^M, \quad j\in \EN_d\label{eq:prRef4}\\
 &v_i = l_ip_i^2, \quad i\in \EN_g\label{eq:prRef5},
\end{align}
\end{subequations}
where $\bv\in\R^{N_g}$ is the vector with coordinates $v_i$, $i\in\EN_g$. Thus, the strategy $p_i$ of each generator is augmented by the auxiliary parameter $v_i$.
However, the problem~\eqref{eq:problemRef0} is still non-convex due to the constraints in \eqref{eq:prRef5}. Following the idea in \cite{Zhao2017}, we present a new reformulation, where each non-convex equality constraint is replaced by the corresponding convex inequality one. Thus, we obtain the convex optimization problem
\begin{equation}\label{eq:problem_Ref}
 \min_{\bp,\bv} \, \sum_{i\in \EN_g} C_i(p_i) - \sum_{j\in\EN_d} U_j(p_j)
\end{equation} \addtocounter{equation}{-1}\vspace{-0.3cm}
\begin{subequations}
 \begin{align}
s.t. \, &\sum_{i\in \EN_g}(p_i - v_i) = \sum_{j\in\EN_d} p_j\label{eq:pr_Ref2}\\
 &p_i^m\le p_i\le P_i^M, \quad i\in \EN_g\label{eq:pr_Ref3}\\
 &p_j^m\le p_j\le P_j^M, \quad j\in \EN_d\label{eq:pr_Ref4}\\
 &v_i \ge l_ip_i^2, \quad i\in \EN_g\label{eq:pr_Ref5}.
 \end{align}
\end{subequations}
Next, we establish the relation between the convex problem~\eqref{eq:problem_Ref} and the initial one~\eqref{eq:problem}. This will be done under the following two technical assumptions.
\begin{assumption}\label{assum:bounds}
 The upper and low bounds for the feasible power generation and demand (see \eqref{eq:pr3} and \eqref{eq:pr4}) satisfy
 \[\sum_{j\in\EN_d}P^M_j\ge\sum_{i\in \EN_g}(p^m_i - l_i(p^m_i)^2).\]
\end{assumption}
The assumption above repeats the sufficient condition for an appropriate convex reformulation of the problem~\eqref{eq:problem} presented in \cite{Zhao2017} (see also Remark~1 in \cite{Zhao2017}).
The next assumption is the Slater constraint qualification for the reformulated problem~\eqref{eq:problem_Ref}. It will enable the relation analysis based on the Karush-Kuhn-Tacker conditions for the optimal primal dual pair of the problem~\eqref{eq:problem_Ref}.
\begin{assumption}\label{assum:Slaters}
 There exists a feasible point $(\hat{\bp}, \hat{\bv})$ for the problem~\eqref{eq:problem_Ref} such that
$\sum_{i\in \EN_g}(\hat p_i - \hat v_i) = \sum_{j\in\EN_d} \hat p_j$, $p_i^m<\hat p_i< P_i^M$, $i\in \EN_g$, $p_j^m<\hat p_j< P_j^M$,  $j\in \EN_d$, $\hat v_i > l_i\hat p_i^2 $, $i\in \EN_g$.
\end{assumption}
The next proposition states the desired relation between the initial problem~\eqref{eq:problem} and the convex one~\eqref{eq:problem_Ref} above.
\begin{prop}\label{prop:equiv}
  Let Assumptions~\ref{assum:bounds} and \ref{assum:Slaters} hold. Then any solution to the  problem~\eqref{eq:problem_Ref} is a solution to the problem~\eqref{eq:problem}.
\end{prop}
\begin{proof}
 See Appendix.
\end{proof}

The optimization problem \eqref{eq:problem_Ref} can be considered a particular case of the general distributed optimization problem \eqref{eq:const_opt}.
Indeed, let $\zb = (\bp, \bv)$ be the vector of joint strategies of the generators and responsive demands in the network, $F_i(\zb)=C_i(p_i)$, if $i\in\EN_g$, $F_j(\zb)=-U_i(p_j)$, if $j\in\EN_g$, and the constraints \eqref{eq:pr_Ref2}-\eqref{eq:pr_Ref5} be distributed over these agents as follows:
\begin{align*}
 LC_i = \{&c_i^1(\zb) = p_i - P_i^M\le 0, c_i^2(\zb) = p_i^m-p_i\le 0, \cr
 &c_i^3(\zb)=\sum_{i\in \EN_g}(p_i - v_i) - \sum_{j\in\EN_d} p_j\le 0,\cr
 &c_i^4(\zb)=-\sum_{i\in \EN_g}(p_i - v_i) + \sum_{j\in\EN_d} p_j\le 0,\cr
 &c_i^5(\zb) = (l_ip_i^2-v_i)1_{\{p_i\in[p_i^m,P_i^M]\}}\cr
 &\hspace{1.2cm}+(2l_ip_i^m-v_i)1_{\{p_i<p_i^m\}} \cr
 &\hspace{1.2cm}+(2l_ip_i^M-v_i)1_{\{p_i>p_i^M\}}\le 0\}, \cr
 &\hspace{3.5cm} \mbox{ if $i\in\EN_g$},\cr
 LC_j = \{&c_j^1(\zb) = p_j - P_j^M\le 0, c_j^2(\zb) = p_j^m-p_j\le 0\}, \cr
 &\hspace{3.5cm} \mbox{ if $j\in\EN_d$}.
\end{align*}
Note that the condition $c_i^5(\zb)\le0$ for $i\in\EN_g$ above corresponds to the constraint \eqref{eq:pr_Ref5}. We modified this constraint without changing the problem (due to existence of the hard constraints \eqref{eq:pr_Ref3}) to be able to use the result from Theorem~\ref{th:main_ps} requiring bounded gradients of the constraint functions.
Given the properties of the cost and utility functions (see~\eqref{eq:costfun} and \eqref{eq:utfun}), the objective function $F(\zb)=\sum_{i\in\EN_g}F_i(\zb)+\sum_{j\in\EN_d}F_j(\zb)$ is strongly convex and, hence, inf-compact. Moreover, Assumptions~\ref{assum:bound_grad} and ~\ref{assum:Lipschitz} hold for the gradients of the functions $F_i$, $F_j$ and local constraint functions $c_i^1$, $c_i^2$, $c_i^3$, $c_i^4$, $c_i^5$, $c_j^1$, $c_j^2$, $i\in\EN_g$, $j\in\EN_d$.
Thus, the problem \eqref{eq:problem_Ref} is equivalent to
\begin{align}\label{eq:fref}
 \min_{\zb} &\, F(\zb)=\sum_{i\in\EN_g}F_i(\zb)+\sum_{j\in\EN_d}F_j(\zb),\cr
s.t. \, &\zb\in \left(\cap_{i\in\EN_g} LC_i\right)\bigcap \left(\cap_{j\in\EN_d} LC_j\right).
 \end{align}
and the following result can be formulated
\begin{cor}\label{cor:em}
 Let Assumptions~\ref{assum:bounds} and \ref{assum:Slaters} hold for the problem \eqref{eq:problem}. Then under an appropriate choice of the parameters $a_t$, $r_t$ (see Assumption~\ref{assum:Slaters}), the penalty-based push-sum algorithm~\eqref{eq:constrps} applied to the reformulated problem~\eqref{eq:fref} converges to the optimal solution to \eqref{eq:problem} as time tends to infinity.
\end{cor}

\subsection{Simulation Results}

In this section we will substantiate our theoretic result stated in Corollary~\ref{cor:em} for the energy management problem \eqref{eq:problem} with simulations. For this purpose, we verify Theorem~\ref{th:main_ps} by comparing the optimum $\bm{p}^*$ of the problem \eqref{eq:problem} with the iterations of the penalty-based push-sum algorithm presented in Section~\ref{sec:ps}.

For our simulation we use a small setup of two generator and two consumer nodes. The cost functions of the generators are as in \eqref{eq:costfun}, whereas the utility functions of the demand nodes are
\begin{equation}
U_j(p_j) =  \begin{cases}
\begin{array}{cc}
\omega_jp_j - \alpha_jp_j^2, &  p_j \leq \frac{\omega_j}{2K_j\alpha_j} \\
(\omega_j - \frac{\omega_j}{K_j})p_j-\frac{\omega^2_j}{4K_j^2\alpha_j}, & p_j > \frac{\omega_j}{2K_j\alpha_j}
\end{array} \end{cases},
\end{equation}
where $K_j<1$ is a positive constant for each $j\in\EN_d$. Thus, the properties \eqref{eq:utfun} are met.
For the parameters for the cost functions, the lower and upper bounds on $p_i$ and $p_j$, as well as for the transmission loss coefficients $l_i$, $i\in\EN_g$, we rely on settings in \cite{Zhao2017}.
We model our time-varying communication architecture with a changing signal $s_t$ that chooses the current graph $G(s(t))$ sequentially from the set $\mathcal{G} = \left \lbrace G_1, G_2 \right \rbrace$, where $G_1$ and $G_2$ are not strongly connected but their union is. The communication architecture is depicted in Figure~\ref{fig:commarch}.
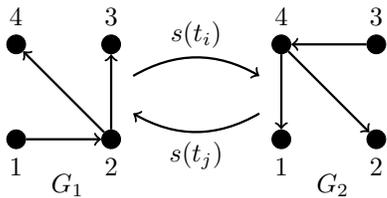
\begin{figure}[h!]
	\centering

	\begin{tikzpicture}
	
	\tikzstyle{mybox} = [draw=black, fill=white!10, very thick,
	rectangle, rounded corners, inner sep=10pt, inner ysep=10pt]
	
	\tikzstyle{boxdec} = [draw=black, fill=yellow!10, very thick,
	rectangle, rounded corners, inner sep=10pt, inner ysep=10pt]
	\tikzstyle{box} = [draw=black, fill=yellow!10, very thick,
	rectangle, rounded corners, inner sep=10pt, inner ysep=10pt]
	
	\tikzstyle{gnode} = [draw,shape=circle,fill=black,  inner sep=2.5pt]

	\node[gnode, label=below:{$1$}] at (0,0) (n1){};
	\node[gnode, right =  of n1, label=below:{$2$}](n2){};
	\node[gnode,  above = of n2, label=above:{$3$}](n3){};
	\node[gnode,  left = of n3, label=above:{$4$}](n4){};

	\node[gnode, right= 2cm of n2, label=below:{$1$}]  (g1){};
	\node[gnode, right =  of g1, label=below:{$2$}](g2){};
	\node[gnode,  above = of g2, label=above:{$3$}](g3){};
	\node[gnode,  left = of g3, label=above:{$4$}](g4){};

	\node[coordinate, above right = 0.75 and 0.2 of n2](loop1) {};
	\node[coordinate, above left = 0.75 and 0.2 of g1](loop2) {};
	
	\node[coordinate, above right = 0.25 and 0.2 of n2](loop3) {};
	\node[coordinate, above left = 0.25 and 0.2 of g1](loop4) {};

	\node[below right = 0.25 and 0.25 of n1](loop5) {$G_1$};
	\node[below right = 0.25 and 0.25 of g1](loop6) {$G_2$};

	\draw[->, line width=0.3mm] (n1) -- (n2);
	\draw[->, line width=0.3mm] (n2) -- (n3);
	\draw[->, line width=0.3mm] (n2) -- (n4);
	
	\draw[->, line width=0.3mm] (g4) -- (g2);
	\draw[->, line width=0.3mm] (g4) -- (g1);
	\draw[->, line width=0.3mm] (g3) -- (g4);

	\path[->, line width=0.3mm]
	(loop1) edge[bend left] node [above] {$s(t_i)$} (loop2)
	(loop4) edge[bend left] node [below] {$s(t_j)$} (loop3);

	\end{tikzpicture}
	
	\caption{Time-dependent directed communication architecture.}
	\label{fig:commarch}
\end{figure}

The results of our simulations are shown in Table~\ref{fig:compareproblems} and Figure \ref{fig:errorpushsum4nodesp}. The table states the optimal values $\bm{p}^*$ of the generators and consumers together with the result of the push-sum algorithm after $3\times 10^4$  iterations.
The figure shows the convergence of the relative error $\mbox{error}_k = \frac{|p_k-p_k^*|}{p_k^*}$, $k\in[4]$, to zero as time runs.
We can notice that the relative error at the demand nodes ($k=3,4$) approaches $0$ already after $500$ iterations, whereas the generator's errors need significantly more time to get close to $0$. This effect is due to a more complex structure of generators' constraints, for which an optimal choice of penalty parameters needs to be studied in the future work.

\renewcommand{\arraystretch}{1.5}
\begin{table}[h!]
	\centering
	\begin{tabular}{c|c c} \hline
		        &       $\bm{p}^*$       &  $\bm{{p}}$ \\ \hline
		$p_1^G$ &            81.98             &                   81.69              \\ \hline
		$p_2^G$ &            124.80            &                  122.75                   \\ \hline
		$p_1^D$ &            100.34            &                   99.61                  \\ \hline
		$p_2^D$ &            100.00           &                   99.64                   \\ \hline
	\end{tabular}
	\caption{Optimal $\bm{p}^*$ vs algorithm output after $3\times 10^4$ iterations.}
	\label{fig:compareproblems}
\end{table}

\begin{figure}[h!]
	\centering
	\includegraphics[width=0.9\linewidth]{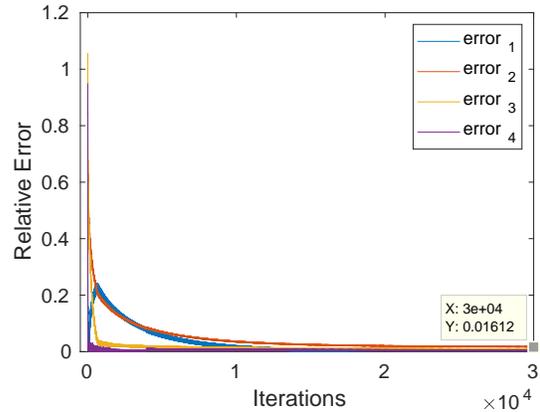}
	\caption{Convergence of the relative error $\mbox{error}_k$, $k\in[4]$, between optimum $p_k^*$ and algorithm output $p_k$.}
	\label{fig:errorpushsum4nodesp}
\end{figure}

\section{Conclusion}\label{sec:conclusion}
In this paper we extended the distributed push-sum algorithm to the case of constrained convex optimization. The penalty-based push-sum algorithm was presented and its convergence to a system's optimum was proven. We demonstrated applicability of the proposed procedure to distributed energy management in smart grid.
The future work will focus on such questions as the convergence rate of the penalty-based push-sum algorithm and its dependence on the communication topology as well as an optimal choice of penalty functions and penalty parameters.

\bibliographystyle{plain}


\begin{thebibliography}{10}

\bibitem{Bertsekas}
D.P. Bertsekas.
\newblock {\em Nonlinear Programming}.
\newblock Athena Scientific, 1999.

\bibitem{16A2}
D.~Kempe, A.~Dobra, and J.~Gehrke.
\newblock Gossip-based computation of aggregate information.
\newblock pages 482--491. IEEE Computer Society, 2003.

\bibitem{Zhao23}
J.~{Mohammadi}, G.~{Hug}, and S.~{Kar}.
\newblock Asynchronous distributed approach for dc optimal power flow.
\newblock In {\em 2015 IEEE Eindhoven PowerTech}, pages 1--6, June 2015.

\bibitem{A1}
A.~Nedi\'c and A.~Olshevsky.
\newblock Distributed optimization over time-varying directed graphs.
\newblock {\em {IEEE} Trans. Automat. Contr.}, 60(3):601--615, 2015.

\bibitem{Neglia}
G.~Neglia, G.~Reina, and S.~Alouf.
\newblock Distributed gradient optimization for epidemic routing: a preliminary
  evaluation.
\newblock In {\em Proc. of 2nd IFIP Wireless Days 2009}, December 2009.

\bibitem{PflugPenF}
G.~Pflug.
\newblock On the convergence of a penalty-type stochastic optimization
  procedure.
\newblock {\em Journal of Information and Optimization Sciences},
  2(3):249--258, 1981.

\bibitem{Zhao2}
M.~{Pipattanasomporn}, H.~{Feroze}, and S.~{Rahman}.
\newblock Multi-agent systems in a distributed smart grid: Design and
  implementation.
\newblock In {\em 2009 IEEE/PES Power Systems Conference and Exposition}, pages
  1--8, March 2009.

\bibitem{Rabbat}
M.~Rabbat and R.~Nowak.
\newblock Distributed optimization in sensor networks.
\newblock In {\em Proceedings of the 3rd International Symposium on Information
  Processing in Sensor Networks}, IPSN '04, pages 20--27, New York, NY, USA,
  2004. ACM.

\bibitem{Zhao26}
N.~{Rahbari-Asr}, U.~{Ojha}, Z.~{Zhang}, and M.~{Chow}.
\newblock Incremental welfare consensus algorithm for cooperative distributed
  generation/demand response in smart grid.
\newblock {\em IEEE Transactions on Smart Grid}, 5(6):2836--2845, Nov 2014.

\bibitem{Asr2014}
N.~Rahbari-Asr, U.~Ojha, Z.~Zhang, and M.~Chow.
\newblock Incremental welfare consensus algorithm for cooperative distributed
  generation/demand response in smart grid.
\newblock {\em IEEE Transactions on Smart Grid}, 5(6):2836--2845, Nov 2014.

\bibitem{Zhao9}
W.~Ren and R.~W. {Beard}.
\newblock Consensus seeking in multiagent systems under dynamically changing
  interaction topologies.
\newblock {\em IEEE Transactions on Automatic Control}, 50(5):655--661, May
  2005.

\bibitem{robbins1985convergence}
H.~Robbins and D.~Siegmund.
\newblock A convergence theorem for non negative almost supermartingales and
  some applications.
\newblock In {\em Herbert Robbins Selected Papers}, pages 111--135. Springer,
  1985.

\bibitem{TAC2017}
T.~{Tatarenko} and B.~{Touri}.
\newblock Non-convex distributed optimization.
\newblock {\em IEEE Transactions on Automatic Control}, 62(8):3744--3757, Aug
  2017.

\bibitem{touri2010infinite}
B.~Touri and A.~Nedi{\'c}.
\newblock When infinite flow is sufficient for ergodicity.
\newblock In {\em 49th IEEE Conference on Decision and Control (CDC)}, pages
  7479--7486. IEEE, 2010.

\bibitem{Tsianos2012}
K.I. Tsianos, S.~Lawlor, and M.G. Rabbat.
\newblock {Consensus-based distributed optimization: Practical issues and
  applications in large-scale machine learning}.
\newblock In {\em 2012 50th Annual Allerton Conference on Communication,
  Control, and Computing (Allerton)}, pages 1543--1550. IEEE, October 2012.

\bibitem{WillertS16}
V.~Willert and M.~Schnaubelt.
\newblock Data control: A feedback control design for clustering of networked
  data (in german).
\newblock {\em Automatisierungstechnik}, 64(8):618--632, 2016.

\bibitem{Zhao2017}
C.~Zhao, J.~He, P.~Cheng, and J.~Chen.
\newblock Consensus-based energy management in smart grid with transmission
  losses and directed communication.
\newblock {\em IEEE Transactions on Smart Grid}, 8(5):2049--2061, Sep. 2017.

\end{thebibliography}

\appendix

\emph{Proof of Proposition~\ref{prop:equiv}}
\begin{proof}
Let us consider the Lagrangian function defined for the convex problem~\eqref{eq:problem_Ref}, namely
\begin{align*}
 &L(\bp,\bv,\lambda,\bmu,\bgamma,\btheta) = \sum_{i\in \EN_g} C_i(p_i) - \sum_{j\in\EN_d} U_j(p_j)\cr
 &+\lambda(\sum_{i\in \EN_g}(p_i - v_i) - \sum_{j\in\EN_d} p_j)+\sum_{i\in\EN_g}\mu_i(l_ip_i^2 - v_i)\cr
 &+\sum_{k\in[N]}\gamma_k(p^m_k - p_k) + \sum_{k\in[N]}\theta_k(p_k - P^M_k),
\end{align*}
where  $\bmu$, $\bgamma$, $\btheta$ are the vectors of the Lagrangian multipliers with corresponding dimensions.
As the problem~\eqref{eq:problem_Ref} is convex and Assumption~\ref{assum:Slaters} holds, we can use necessary and sufficient Karush-Kuhn-Tacker conditions for $[(\bp^*,\bv^*), (\lambda^*,\bmu^*,\bgamma^*,\btheta^*)]$ being an optimal primal dual pair\footnote{Note that existence of an optimal primal dual is guaranteed in this case as well (see, for example, Proposition 5.3.1 in \cite{Bertsekas}).}.
Thus,
\begin{subequations}\label{eq:KKT}
\begin{align}
 &\frac{\partial L}{\partial p_i} = \frac{d C_i(p_i^*)}{d p_i} + \lambda^* + 2\mu_i^*l_ip_i^* - \gamma_i^* + \theta_i^* = 0, \, i\in\EN_g, \label{eq:KKT1}\\
 &\frac{\partial L}{\partial p_j} = -\frac{d U_j(p_j^*)}{d p_j} - \lambda^*  - \gamma_j^* + \theta_j^* = 0, \, j\in\EN_d, \label{eq:KKT2}\\
 &\frac{\partial L}{\partial v_i} = - \lambda^* + \mu_i^*= 0, \, i\in\EN_g, \label{eq:KKT3}\\
 &\gamma_k^*(p_k^m-p_k^*) = 0, \, \gamma_k^*\ge 0, \, k\in[N], \label{eq:KKT4}\\
 &\theta_k^*(p_k^*-P^M_k) = 0, \, \theta_k^*\ge 0, \, k\in[N], \label{eq:KKT5}\\
 &\mu_i^*(l_i(p_i^*)^2-v_i^*) = 0, \, \mu_i^*\ge 0, \, i\in\EN_g\label{eq:KKT6}.
\end{align}
\end{subequations}
Suppose that
\begin{align}\label{eq:suppose}
 v_{i'}^* > l_{i'}(p_{i'}^*)^2 \quad \mbox{for some $i'\in\EN_g$}.
\end{align}
Then, due to \eqref{eq:KKT6}, $\mu_{i'}^*=0$. Hence, according to \eqref{eq:KKT3}, $\lambda^* = 0$ and $\mu_{i}^*=0$ for all $i\in\EN_g$.
Next, since Assumption~\ref{assum:Slaters} guarantees that $p_k^m\ne P_k^M$ for all $k\in[N]$, $\gamma_k^*=0$ or $\theta_k^*=0$ for all $k\in[N]$ (see \eqref{eq:KKT4} and \eqref{eq:KKT5}). Let us consider any $j\in\EN_d$. Due to the fact that $\frac{d U_j(p_j^*)}{d p_j}>0$ (see \eqref{eq:utfun}), $\gamma_j^*\ge 0$ (see \eqref{eq:KKT4}), and condition \eqref{eq:KKT2}, we conclude that $\gamma_j^*=0$ and $\theta_j^*>0$. Thus, $p_j^*=P^M_j$ (see \eqref{eq:KKT5}) for all $j\in\EN_d$. Analogously, for any $i\in\EN_g$, due to the property of the cost functions, namely $\frac{d C_i(p_i^*)}{d p_i}>0$ (see \eqref{eq:costfun}), we obtain that $\gamma_i^*>0$ and $\theta_i^*=0$ (see \eqref{eq:KKT1}). Thus, $p_i^*=p^m_i$ (see \eqref{eq:KKT4}) for all $i\in\EN_g$.

Next, taking into account \eqref{eq:suppose} and the feasibility conditions \eqref{eq:pr_Ref2} and \eqref{eq:pr_Ref5}, we get
\begin{align}\label{eq:contr}
  \sum_{j\in\EN_d}P^M_j&=\sum_{j\in\EN_d} p_j^*  = \sum_{i\in \EN_g}(p^*_i - v^*_i)\cr
 &<\sum_{i\in \EN_g}(p^m_i - l_i(p^m_i)^2),
\end{align}
which contradicts Assumption~\ref{assum:bounds}. Thus, \eqref{eq:suppose} cannot hold, which implies that $v_{i}^* = l_{i}(p_{i}^*)^2$ for all $i\in\EN_g$.
Hence, the optimal solution $(\bp^*, \bv^*)$ to the problem \eqref{eq:problem_Ref} necessarily satisfies the feasibility conditions of the problem \eqref{eq:problemRef0}, which is equivalent to the initial problem \eqref{eq:problem}. By noticing that the objective function in the optimization problems is the same, we conclude the proof.
\end{proof}
\bigskip

\emph{Supporting Theorems}

\begin{theorem}~\cite{A1}\label{th:th2}
Consider  the  sequences $\{\zb_i(t)\}_t$, $i\in [n]$, generated  by the algorithm~\eqref{eq:constrps}. Assume that the graph sequence
$\{G(t)\}$ is $B$-strongly connected and Assumptions~\ref{assum:bound_grad} and \ref{assum:parameters} hold.

(a) Then  $\lim_{t\to\infty}\|\zb_i(t+1)-\bar{\zx}(t) \|=0$ for all $i\in[n]$.

Moreover,

(b) If $\{b_t\}$ is a non-increasing positive scalar sequence with  $\sum_{t=1}^{\infty}b_ta_t\|\fb_i(\zb_i(t+1))+r_t\bpsi_i(\zb_i(t+1))\|_1<\infty$  for all $i\in [n],$ then
$\sum_{t=0}^{\infty}b_{t}\left\|\zb_i(t+1)-\bar{\zx}(t) \right\|<\infty$ for all $i$,
where $\|\cdot\|_1$ is the $l^1$-norm in $\mathbb{R}^d$.
\end{theorem}

The next theorem is the well-known result on non-negative variables \cite{robbins1985convergence}.

\begin{theorem}\label{th:th_nonnegrv}
 Let $z_n, \beta_n, \xi_n,$ and $\zeta_n$ be non-negative  variables such that
 \begin{align*}
  z_{n+1}\le z_n(1+\beta_n)-\zeta_n+\xi_n.
 \end{align*}
Then $\lim_{n\to\infty} z_n$ exists and is finite and $\sum_{n=1}^{\infty}\zeta_n<\infty$  on $\{\sum_{n=1}^{\infty}\beta_n<\infty, \sum_{n=1}^{\infty}\xi_n<\infty\}$.
\end{theorem}

\end{document}